\documentclass[runningheads]{llncs}
\usepackage{graphicx}
\expandafter\let\csname proof\endcsname\relax
\expandafter\let\csname endproof\endcsname\relax

\usepackage{amsmath}
\usepackage{amsfonts}
\usepackage{amssymb}
\usepackage{mathrsfs}
\usepackage{bm}
\usepackage{subfigure}
\usepackage{booktabs}
\usepackage{caption}
\usepackage{hyperref} 
\usepackage{float}
\usepackage{adjustbox}

\usepackage{color}
\usepackage{amssymb}
\usepackage{amsmath}
\usepackage{amsthm}
\usepackage{graphicx}
\usepackage[ruled,linesnumbered]{algorithm2e}
\usepackage{algorithmic}
\usepackage{subfigure}
\usepackage{cite}
\usepackage{appendix}
\usepackage{booktabs}
\usepackage{multirow}
\usepackage{array}
\usepackage{bm}
\usepackage{newtxtext,newtxmath}
\usepackage{makecell}
\usepackage{indentfirst}
\usepackage{authblk}

\newcommand{\tgao}[1]{{\color{blue}{[tgao: ##]}}}

\begin{document}
\title{Privacy-Preserving Discretized Spiking Neural Networks}
%
\titlerunning{Privacy-Preserving Discretized Spiking Neural Networks}
%
\author{Pengbo Li\inst{1} \and
Ting Gao\inst{1}\thanks{Corresponding author: tgao0716@hust.edu.cn} \and
Huifang Huang\inst{2} \and
Jiani Cheng\inst{1} \and
Shuhong Gao\inst{3} \and
Zhigang Zeng\inst{4} \and
Jinqiao Duan\inst{5}}
%

\institute{Center for Mathematical Sciences, Huazhong University of Science and Technology, Wuhan 430074, China \and School of Mathematics and Statistics, Huazhong University of Science and Technology, Wuhan 430074, China \and School of Mathematical and Statistical Sciences, Clemson University, Clemson SC 29634-0975 U.S.A \and School of Artificial Intelligence and Automation, Huazhong University of Science and Technology, Wuhan 430074, China \and  
Greater Bay University, Dongguan 523830, China}
\maketitle    

\begin{abstract}

The rapid development of artificial intelligence has brought considerable convenience, yet also introduces significant security risks. One of the research hotspots is to balance data privacy and utility in the real world of artificial intelligence. The present second-generation artificial neural networks have made tremendous advances, but some big models could have really high computational costs. The third-generation neural network, SNN (Spiking Neural Network), mimics real neurons by using discrete spike signals, whose sequences exhibit strong sparsity, providing advantages such as low energy consumption and high efficiency. In this paper, we construct a framework to evaluate the homomorphic computation of SNN named FHE-DiSNN that enables SNN to achieve good prediction performance on encrypted data. First, benefitting from the discrete nature of spike signals, our proposed model avoids the errors introduced by discretizing activation functions. Second, by applying bootstrapping, we design new private preserving functions \textbf{FHE-Fire} and \textbf{FHE-Reset}, through which noise can be refreshed, allowing us to evaluate SNN for an arbitrary number of operations. Furthermore, We improve the computational efficiency of FHE-DiSNN while maintaining a high level of accuracy. Finally, we evaluate our model on the MNIST dataset. The experiments show that FHE-DiSNN with 30 neurons in the hidden layer achieves a minimum prediction accuracy of 94.4\%. Under optimal parameters, it achieves a 95.1\% accuracy, with only a 0.6\% decrease compared to the original SNN (95.7\%). These results demonstrate the superiority of SNN over second-generation neural networks for homomorphic evaluation.

\end{abstract}

\section{Introduction}

\textbf{Privacy-Preserved AI.}  Machine learning algorithms based on deep neural networks have attracted extensive attention as a key technology in Artificial Intelligence (AI). These achievements have been widely applied in various fields such as image processing, intelligent transportation, and security. However, users face challenges of insufficient local computing power when training neural network models with a large number of parameters, which leads to the consideration of MLaaS(Machine Learning as a Service)\cite{ribeiroMLaaSMachineLearning2015} to outsource the computation of neural network models to cloud services. However, outsourcing brings risks of data security breaches. To address this issue, many privacy protection techniques are applied to machine learning models, such as homomorphic encryption (HE), differential privacy (DP), and secure multi-party computation (SMC) based on cryptography.

Homomorphic encryption refers to the ability to perform arbitrary computations on ciphertext without decryption. This unique property enables homomorphic encryption to have broad theoretical and practical applications, such as secure encrypted retrieval in cloud computing and secure multi-party computation. Therefore, researching homomorphic encryption holds significant scientific and practical value. In 2009, Gentry\cite{gentryFullyHomomorphicEncryption2009,gentryComputingArbitraryFunctions2010a} constructs the first fully homomorphic encryption (FHE) scheme, which is a major breakthrough in the field of cryptography.
So far, there have been four generations of FHE. In the first generation \cite{gentryComputingArbitraryFunctions2010a}, Gentry constructs a true bootstrapping process, although its practical performance is poor. The second-generation scheme, represented by BFV \cite{brakerskiLeveledFullyHomomorphic2014} and BGV \cite{fanSomewhatPracticalFully2012a}, introduces a technique called modulus reduction, which builds leveled HE schemes that can compute addition and multiplication of predefined depth. Another advantage of the second-generation scheme is the SIMD operation, allowing parallel processing of thousands of plaintexts in corresponding ciphertext slots, greatly improving the scheme's performance. 
CKKS \cite{cheonHomomorphicEncryptionArithmetic2017} is a modification of BFV schemes that supports homomorphic real number operations with fixed precision. The third-generation schemes include FHEW \cite{ducasFHEWBootstrappingHomomorphic2015}, TFHE \cite{chillottiTFHEFastFully2020} and Gao et al.\cite{gaoEfficientFullyHomomorphic2018, caseFullyHomomorphicEncryption2019} that have fast bootstrapping and enable an unlimited number of operations.

Although there are many works based on the early second-generation FHE, it only supports homomorphic operations of addition and multiplication, while practical computations often involve non-linear operations such as comparison and maximization, especially activation functions in neural networks. To address these issues, Gilad-Bachrach et al. \cite{gilad-bachrachCryptoNetsApplyingNeural2016} propose the CryptoNets method, which replaces non-linear activation functions with polynomial functions. However, polynomials of a high degree are needed for a good approximation of nonlinear functions used in machine learning. Mohassel et al. \cite{mohasselSecureMLSystemScalable2017} introduce an interactive algorithm that utilizes two servers to handle non-linear function problems, but it requires continuous interaction between the client and the servers, leading to high communication costs. Chabanne et al. \cite{chabannePrivacyPreservingClassificationDeep2017} modify the model for the prediction phase to address the non-linear activation function problem, but this approach results in a loss of precision in prediction and training results.

In \cite{bourseFastHomomorphicEvaluation2018}, the authors design FHE-DiNN, a discrete neural network framework based on the third-generation TFHE \cite{chillottiTFHEFastFully2020} scheme, where the output of each neuron is refreshed through bootstrapping, enabling homomorphic computation for arbitrary depths of networks. Unlike standard neural networks, FHE-DiNN utilizes a discretized neural network that restricts the propagated signals to integers and employs the sign function as the activation function to achieve scale invariance. FHE-DiNN exhibits fast computation speed but has lower model prediction accuracy. This work inspires us to consider whether SNN neurons that naturally output 0 and 1 binary values can also be efficiently homomorphically evaluated.
\\
\\
\textbf{Spiking Neural Network.}  Compared to other neural network models, Spiking Neural Networks (SNN) are generally more reliable in biological interpretation. As the third generation of neural networks\cite{Maass1996NetworksOS}, SNNs have gained increasing attention due to their rich spatiotemporal neural dynamics, diverse coding mechanisms, and low-power advantages in neuromorphic chips.

In contrast to the prosperity of artificial neural networks (ANNs), the development of SNNs is still in the early stage. Currently, researches in SNNs mainly focus on five major directions: neuron models, training algorithms, programming frameworks, datasets, and hardware chips. In response to the dynamic characteristics of the potential of neurons, neurophysiologists have constructed many models. These models are the basic units that make up spiking neural networks and determine the basic dynamic characteristics of the network. Among them, the most influential models include the Hodgkin-Huxley (H-H) model \cite{Hodgkin1952AQD}, the leaky integrate-and-fire (LIF) model \cite{Wu2017SpatioTemporalBF}, the Izhikevich model \cite{Izhikevich2003SimpleMO}, and the spike response model \cite{jolivetSpikeResponseModel2003}(SRM), etc.

The training algorithms of SNNs can be mainly divided into three types: (1) gradient-free training algorithms represented by spike-timing dependent plasticity (STDP) \cite{Izhikevich2007SolvingTD}; (2) direct conversion of ANNs to SNNs; (3) gradient-surrogate training algorithms represented by error back-propagation in the spatiotemporal domain. Bohte et al.\cite{Boht2000ErrorbackpropagationIT} first propose a gradient descent learning algorithm that can be applied to multi-layer feed-forward spiking Neural networks, called the SpikeProp learning algorithm. Recently, Wu et al.\cite{Wu2017SpatioTemporalBF}propose the spatiotemporal back propagation (STBP) method for the direct training of SNNs, and significantly improve it in order to be compatible with a much deeper structure, larger dataset, and better performance. 

Considering the superior stability and lower energy consumption of SNNs in handling discrete data, it is reasonable to explore the integration of SNNs with FHE. The advantage of FHE-DiSNN lies in its strong robustness to discretization. In the process of converting traditional ANNs to homomorphic computation, discretizing the activation function is a challenging problem. Whether it is approximating with low-degree polynomials\cite{gilad-bachrachCryptoNetsApplyingNeural2016} or directly setting it as the sign function (in DiNN), both methods result in a loss of accuracy. SNN, on the other hand, naturally avoids this problem since all its outputs are binary pulse signals taking values from {0,1}. This property also satisfies the scale-invariant property, eliminating the need to consider the influence of computation depth when designing the discretization. Inspired by FHE-DiNN, we also provide discretization methods for linear neuron models such as LIF and IF and prove that the discretization error caused by this method is very small.
\\
\\
\textbf{Our Contribution.}
In this paper, we construct a novel framework called FHE-DiSNN with the following benefits: 
\begin{itemize}
    \item[$\bullet$] develop a low-accuracy-loss method to discretize SNN to DiSNN with controllable error.
    \item[$\bullet$] design new private preserving functions \textbf{FHE-Fire} and \textbf{FHE-Reset} with TFHE bootstrapping technology so that the resulting FHE-DiSNN constructed from DiSNN can have an arbitrary number of operations.
    \item[$\bullet$] propose an easy-extended framework(SNN $\rightarrow$ DiSNN $\rightarrow$ FHE-DiSNN) that allows the prediction procedure of SNN to be evaluated homomorphically.
\end{itemize}

Our experiments on the MNIST\cite{dengMNISTDatabaseHandwritten2012} dataset confirm the advantages of the FHE-DiSNN. First, we train a fully connected SNN with a single hidden layer consisting of 30 neurons. This SNN is constructed based on the IF(Integrate-and-Fire) neuron model and implemented using the Spikingjelly\cite{SpikingJelly} Python package. Then, we convert it to DiSNN with the optimal parameters determined experimentally. The experiments show that DiSNN achieves a prediction accuracy of 95.3\% on plaintext data, with only a marginal decrease of 0.4\% compared to the original SNN's accuracy of 95.7\%. Finally, the accuracy of FHE-DiSNN is evaluated on ciphertext using the TFHE library, resulting in an accuracy rate of 95.1\%. This demonstrates a slight degradation (0.2\%, 0.6\%) compared to both DiSNN (95.3\%) and SNN (95.7\%).
\\
\\
\textbf{Outline of the paper.}
The paper is structured as follows: In Section 2, we provide definitions and explanations of SNN and TFHE, including a brief introduction to the bootstrapping process of TFHE. In Section 3, we present our method of constructing Discretized Spiking Neural Networks and prove that the discretization error can be controlled. In Section 4, we highlight the challenges of evaluating a DiSNN homomorphically and provide a detailed explanation of our proposed solution. In Section 5, we present comprehensive experimental results for verification of our proposed framework. And discuss the challenges and possible future work in section 6.

\section{Preliminary Knowledge}
In this chapter, we commence by presenting the training and prediction methods of the SNN model. Subsequently, we provide a concise introduction to the bootstrapping process in the TFHE scheme.

\subsection{Spiking Neural Network}
The typical structure of a neuron predominantly encompasses three components: dendrites, soma (cell body), and axons. In consideration of the neuron's potential dynamic characteristics during its operation, neurophysiologists have devised diverse models that constitute the foundational constituents of spiking neural networks, thereby exerting influence on the network's fundamental dynamic properties. 

The Hodgkin-Huxley (H-H) model provides a comprehensive and accurate depiction of the intricate electrical activity mechanisms in neurons. However, it entails a complex and extensive system of dynamic equations that impose substantial computational demands, so simplified models remain practical and valuable such as the most widely utilized Leaky Integrate-and-Fire (LIF) model. LIF model simplifies the process of action potentials significantly while retaining three key characteristics: leakage, accumulation, and threshold excitation which are presented below:
\begin{equation}\label{eq1}
\Omega \frac{d V}{d t}=-V+I,
\end{equation}
where $\Omega=R C$ is a time constant, $R$ and $C$ denotes the membrane resistance and capacitance respectively. Building upon the foundation of LIF model, there exist, several variant models, including QIF model \cite{brunelFiringRateNoisy2003}, EIF model \cite{fourcaud-trocmeHowSpikeGeneration2003}, and adaptive EIF model \cite{bretteAdaptiveExponentialIntegrateandFire2005}. Besides, IF model \cite{Abbott1999LapicquesIO} is a further simplification of LIF model, where $\Omega=1$ and $V$ in Equation \ref{eq1} disappear, i.e.$\frac{d V}{d t}=I$.

In practical applications, it is common to utilize discrete difference equations as an approximation method for modeling the equations governing neuronal electrical activity. Although the specific accumulation equations for various neuronal membrane potentials may differ, the threshold excitation and reset equations for the membrane potential remain consistent. Consequently, the neuronal electrical activity can be simplified into three distinct stages: charging, firing, and resetting.
\begin{equation}\label{HSVt}
\begin{aligned}
H[t] & =V_{t-1} + f(V[t-1], I[t]), \\
S[t] & =\textbf{Fire}\left(H[t]-V_{threshold }\right), \\
V[t] &= \textbf{Reset}(H[t]) = \left\{ \begin{aligned} 
V_{reset}, \quad & \text{if} \quad H[t] \geq V_{threshold }, \\ 
H[t], \quad &\text{if} \quad V_{reset} \leq H[t] \leq V_{threshold }, \\ 
 V_{reset}, \quad &\text{if} \quad H[t] \leq V_{reset}.
 \end{aligned} \right.
\end{aligned}
\end{equation}
$\textbf{Fire}(\cdot)$ is a step function:
\begin{equation}
\textbf{Fire}(x) = \left\{ \begin{aligned} 
1, \quad & \text{if} \quad x \geq 0, \\
0, \quad & \text{if} \quad x \le 0. \\
\end{aligned} \right.
\end{equation}
$I[t]$(The subscript $i$ represents the i-th neuron, here we only refer to an arbitrary neuron, so $i$ can be omitted.) represents the total membrane current of the external input from the pre-synaptic neurons. This term can be conceptually interpreted as the voltage increment and mathematically calculated using the equation provided below:
\begin{equation}\label{It}
I[t] = \sum\limits_j w_{ij}S_j[t].
\end{equation}
To mitigate potential confusion, we employ the notation $H[t]$ to denote the membrane potential of the neuron subsequent to the charging phase and prior to spike initiation, while $V[t]$ signifies the membrane potential of the neuron subsequent to spike initiation. The function $f(V[t-1], X[t])$ represents the equation governing the state transition of the neuron, wherein the distinctions between different neuron models manifest in the specific formulation of $f$.

In practical applications, it is common to utilize spike encoding methods to transform image data into appropriate binary inputs format for SNN. Poisson encoding is a commonly used one, in which the inputs are encoded into rate-based spike by the $\lambda$-Poisson process. Additionally, due to the non-differentiable nature of spiking functions, the conventional back-propagation algorithm based on gradient descent in ANNs is not suitable in this context. Therefore, alternative training approaches must be sought. Poisson encoding and surrogate gradient method are utilized in this paper and other  common methodologies for encoding data and training SNN are detailed in Appendix A.

\subsection{Programmable Bootstrapping }
Let $N = 2^k$ and $p >1$ an even integer. Let $Z_p = \{-\frac{p}{2}+1,\dots,\frac{p}{2} \}$ be the ring of integer modulo $p$. Let $X^N + 1$ be the (2N)-th cyclotomic polynomial. Let $q$ be a prime and define $R_{q,N}= R/qR \equiv \mathbb{Z}_q[X] / (X^N + 1) \equiv \mathbb{Z}[X] / (X^N + 1,q)$, similarly for $R_{p,N}$. Vectors are represented by lowercase bold letters, such as $\mathbf{a}$. The $i$-th entry of a vector $\mathbf{a}$ is denoted as $a_i$. The inner product between vectors $\mathbf{a}$ and $\mathbf{b}$ is denoted by $\langle\mathbf{a}, \mathbf{b}\rangle$.  A polynomial $m(X)$ in $R_{p, N}$ corresponds to a message vector of length $N$ over $Z_p$, and the ciphertext for $m(X)$ will be a pair of polynomials in $R_{q, N}$. Detailed fully homomorphic encryption schemes have been included in Appendix B.

When referring to a probability distribution, we indicate that a value $d$ is drawn from the distribution $\mathcal{D}$ as $d \sim \mathcal{D}$.\\

\vspace{-0.3cm}
\begin{theorem}\label{bootthm}
\textup{\textbf{(Programmable bootstrapping\cite{chillottiImprovedProgrammableBootstrapping2021})}}
TFHE/FHEW bootstrapping support the computation of any function $g: Z_{p} \rightarrow Z_{p} \quad \text{and} \quad g(v+\frac{p}{2}) = -g(v) $. We refer to $g$ as the program function of bootstrapping. An LWE ciphertext $LWE_s(m) = (\mathbf{a}, b)$, where $m \in Z_p$, $\mathbf{a} \in Z_p^N$ and $b \in Z_p$, can be bootstrapped into $LWE_s(g(m))$ with very low noise.
\end{theorem}
This process relies on the Homomorphic Accumulator\cite{micciancioBootstrappingFHEWlikeCryptosystems2020} denoted as $ACC_g$.Using the notations of \cite{micciancioBootstrappingFHEWlikeCryptosystems2020}, the bootstrapping process can be broken down into the following steps:

-\textbf{Initialize}: Set the initial polynomial:
\begin{equation}
ACC_g[-b] = X^{-b} \cdot \sum\limits_{i=0}^{N-1} g\left(\left\lfloor \frac{i\cdot p}{2N}\right\rfloor\right)X^{i} \bmod X^N+1.    
\end{equation}

-\textbf{Blind Rotation}: $ACC_g \leftarrow_{+}^{+} -a_i \cdot ek_i$, modifies the content of the accumulator from $ACC_g[-b]$ to $ACC_g[-b+\sum a_i s_i] = ACC_g[-m - e]$, where 
$$\text{ek} = \left(RGSW\left(X^{s_1}\right), \ldots, RGSW\left(X^{s_n}\right)\right) 
,$$
which is a list of materials over $R_q^N$.

-\textbf{Sample Extraction}: $ACC_g = (a(X), b(X))$ is the RLWE ciphertext with component polynomials $a(X) = \sum\limits_{0\le i \le N-1} a_i X^i$ and $b(X) = \sum\limits_{0\le i \le N-1} b_i X^i$. The extraction operation outputs the LWE ciphertext:
$$
RLWE_z \stackrel{\text{Sample Extraction}}{\longrightarrow}LWE_z(g(m)) = (\mathbf{a}, b_0)
,$$
where $\mathbf{a} = (a_0, \ldots, a_{N-1})$ is the coefficient vector of $a(X)$, and $b_0$ is a coefficient of $b(X)$.

-\textbf{Key Switching}: Key switching transforms the LWE instance's key from the original vector $\mathbf{z}$ to the vector $\mathbf{s}$ without changing plaintext message $m$:
$$
LWE_{\mathbf{z}}(g(m)) \stackrel{\text{Key Switching}}{\longrightarrow} LWE_{\mathbf{s}}(g(m)).
$$

Taking a bootstrapping key and a key switching key as input, bootstrapping can be defined as:
\begin{equation}
\text{bootstrapping} = \textbf{KeySwitch} \circ \textbf{Extract} \circ \textbf{BlindRotate} \circ \textbf{Initialize}    
\end{equation}

With program function $g$, bootstrapping takes ciphertext $LWE_s(m)$ as input, and output $LWE_s(g(m))$ with the original secret key $s$:
\begin{equation}
\text{bootstrapping}(LWE_s(m)) = LWE_s(g(m)).    
\end{equation}

This property will be extensively utilized in our context. Since bootstrapping does not alter the secret key, we will use the shorthand $LWE(m)$ to refer to an LWE ciphertext in the rest.

\section{Discretized Spiking Neural Network}
There are two parts to this section. Firstly, we present a simple discretization method to convert SNNs into Discretized Spiking Neural Networks(DiSNNs). We demonstrate that this method guarantees controllable errors for both the IF neuron model and the LIF neuron model. Furthermore, we provide estimations for the extrema of these two discretization models which can be used to determine the size of the plaintext space. Secondly, we propose an efficient method for computing the \textbf{Fire} and \textbf{Reset} functions of the SNN neuron model on the ciphertext, denoted as \textbf{FHE-Fire} and \textbf{FHE-Reset}.

\begin{definition}
A Discretized Spiking Neural Network (DiSNN) is a type of feed-forward spiking neural network in that all weights are discretized into a finite $Z_p$, as well as the inputs and outputs of the neuron model.
\end{definition}

We denote this discretization method as the function:
\begin{equation}
\hat{x} \triangleq \text{Discret}(x, \tau) = \lfloor x \cdot \tau \rceil,
\end{equation}
where $\hat{x}$ represents the value $x$ after discretization and the precision of the discretization can be controlled, with a larger $\tau$ resulting in finer discretization. The equation \ref{HSVt} can be discretized as follows($i$ is omitted like Equation\ref{It}):
\begin{equation}
\begin{aligned}
\hat{I}[t] &= \varSigma  \hat{\omega}_{ij}S_{j}[t], \\
\hat{H}[t] & = \hat{V}[t-1] + f(\hat{V}[t-1],\hat{I}[t]), \\
S[t] & = \textbf{Fire}\left(\hat{H}[t]-\hat{V}_{threshold }\right), \\
\hat{V}[t] &= \textbf{Reset}(\hat{H}[t])= \left\{ \begin{aligned} 
\hat{V}_{reset}, \quad & \text{if} \quad \hat{H}[t] \geq \hat{V}_{threshold }, \\ 
\hat{H}[t], \quad & \text{if} \quad \hat{V}_{reset} \leq \hat{H}[t] \textless \hat{V}_{threshold }, \\ 
\hat{V}_{reset}, \quad & \text{if} \quad \hat{H}[t] \leq \hat{V}_{reset}.
 \end{aligned} \right.
\end{aligned}
\end{equation}
This system of equations clearly shows the advantages of SNN in terms of discretization methods. The binary spike signals with values of 0 and 1 not only avoid the losses incurred by self-discretization but also effectively control the errors caused by discretized weights. The two crucial parameters of SNN, $V_{threshold}$ and $V_{reset}$, are generally set as integers, eliminating any discretization errors. In fact, the only aspect that requires attention is the discretization of weights. An estimate of the upper bound on the discretization error is given in the assertion below.
\begin{proposition} \label{prop1}
For the IF neuron model and LIF neuron model, the discretization error is independent of the scaling factor $\tau$ and only depends on the number of spikes.
\end{proposition}
\begin{proof}
For the IF and LIF neuron models, let a linear function $f$ denote their charging processes. We have,
$$
\begin{aligned}
\tau f(V[t-1], I[t]) = f(\tau V[t-1], \tau I[t]) 
= f(\hat{V}[t-1], \hat{I}[t]).
\end{aligned}
$$
This means that the discretization error is only concentrated in $\hat{I}[t]$, 
$$\begin{aligned}
\max\limits_i| \hat{I}_i[t] - \tau I_i[t] |&= \max\limits_i|\sum \limits_j (\tau w_{ij} - \hat{w}_{ij})S_{j}[t]| \\
&\leq \max \limits_{i,j}|\tau w_{ij} - \hat{w}_{ij}| \cdot |\sum \limits_j S_{j}[t]| \\
&\leq \frac{1}{2} \cdot |\sum \limits_j S_{j}[t]|.
\end{aligned}$$
As the above showing, discretization error is actually independent of $\tau,$  but proportional to the number of spikes.
\end{proof}

Proposition \ref{prop1} provides an upper bound on the overall discretization error, where $\frac{1}{2}$ represents the maximum value of individual weight discretization error. However, in practical situations, not all weights will reach the maximum error. From a mathematical expectation perspective, the discretization error can be further reduced. The proof is provided by the following Proposition.

\begin{proposition}
For the IF and LIF neuron models, assuming the weights follow a uniform distribution on $[-\frac{1}{2},\frac{1}{2}]$ and the number of spikes follows a Poisson distribution with intensity $\lambda$, the mathematical expectation of the discretization error is $\lambda/4$.
\end{proposition}
\begin{proof}
Denote the random variable $\tau w_{ij} - \hat{w}_{ij}$ as $\xi_j(\omega)$, we can see that it follows a uniform distribution on the interval $[-\frac{1}{2}, \frac{1}{2}]$. Set $N(\omega) = \sum_j S_j[t]$ which is a Poisson random variable with intensity $\lambda$. Note that $\mathbb{E}(|\xi_i|) = \frac{1}{4}$, $\mathbb{P}(N=n) = e^\lambda \cdot \frac{\lambda ^n}{n!}$ and $\sum\limits_{n = 0}^\infty \mathbb{P}(N=n) = 1$. Then, the expectation of the error can be written as follows:
\begin{equation*}
	\begin{aligned}
 \mathbb{E}| \hat{I}[t] - \tau I[t]| &\approx \mathbb{E}(\sum \limits_{j=0}^{N(\omega)}|\xi_j|) \\
&= \sum_{n=0}^\infty \mathbb{E}(\sum_{i=0}^n |\xi_i| \mid N(\omega) = n) \cdot \mathbb{P}(N=n) \\
&= \sum_{n=0}^\infty \frac{n}{4} \cdot e^\lambda \cdot \frac{\lambda ^n}{n !} = \frac{\lambda}{4}.
	\end{aligned}
\end{equation*}
\end{proof}

Notice that in the Proof, $\mathbb{E}( |\xi_i| \mid N(\omega) = n) = \mathbb{E}( |\xi_i|)$ is from the independence between $\xi_i(\omega)$ and $N(\omega)$. 

We can obtain a similar conclusion as Proposition \ref{prop1}: the number of spikes, not the parameter $\tau$, affects the magnitude of the error. Although the Proposition above indicates that the size of $\tau$ does not affect the growth of the error, we cannot infinitely increase $\tau$ in order to improve accuracy. This is because larger $\tau$ implies a larger message space, which puts more computational burden on homomorphic encryption. In Proposition \ref{maxmin}, we specify the relationship between them.

\begin{proposition}\label{maxmin}
For the IF and LIF models, the maximum and minimum values generated during the computation process are controlled by $\tau$, the number of spikes, and the extremal values of the weights.
\end{proposition}
\begin{proof}
From equation \ref{HSVt}, for the IF model case, it can be observed that the range of membrane potential $V[t]$ is controlled by the \textbf{Reset} function of the neuron model, bounded within $[V_{reset}, V_{threshold}]$. The maximum and minimum values can only occur in the variable $\hat{H}$. Therefore, the extreme values satisfy the following inequalities:
$$
\begin{aligned}
Max &\triangleq \max(\hat{H}[t]) =\max (\hat{V}[t] + \hat{I}[t]) \\
&\leq \tau (V_{threshold}+ \max\limits_{i,j}(|w_{ij}| \cdot \sum\limits_j |S_j[t]|), \\
Min &\triangleq \min(\hat{H}[t]) \ge - |\hat{V}_{reset}| - |\hat{I}[t]| \\
&\geq -\tau(V_{reset} + \max\limits_{i,j}(|w_{ij}|)\cdot \sum\limits_j |S_j[t]|). \\
\end{aligned}
$$
Besides, we can also prove LIF model in a similar way.
We denote the upper and lower bound as $\alpha, \beta$, respectively.
\end{proof}

\begin{corollary} In general, when the neuron model has $V_{reset} = 0$, the relation between the maximum and minimum values is given by:
\begin{equation}\label{maxminEq}
\beta = -\alpha + \hat{V}_{threshold},
\end{equation}
where $\alpha,\beta$ represent the upper and lower bound of DiSNN from Proposition \ref{maxmin}
\end{corollary}

This seemingly trivial but highly important conclusion ensures the correctness of homomorphic evaluation. We will encounter it in subsequent sections.

\begin{figure*}[htb]
\centering
\subfigure[Finite Filde used in FHE]{
 \label{ZP}
\includegraphics[width=0.36\textwidth]{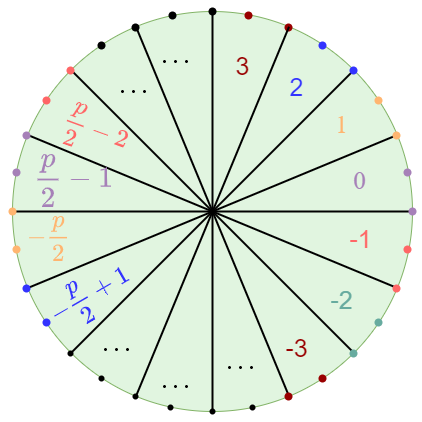}
}
\subfigure[The Prediction Procedure of SNN] {
\label{main_process_figure}
 \includegraphics[width=0.55\textwidth]{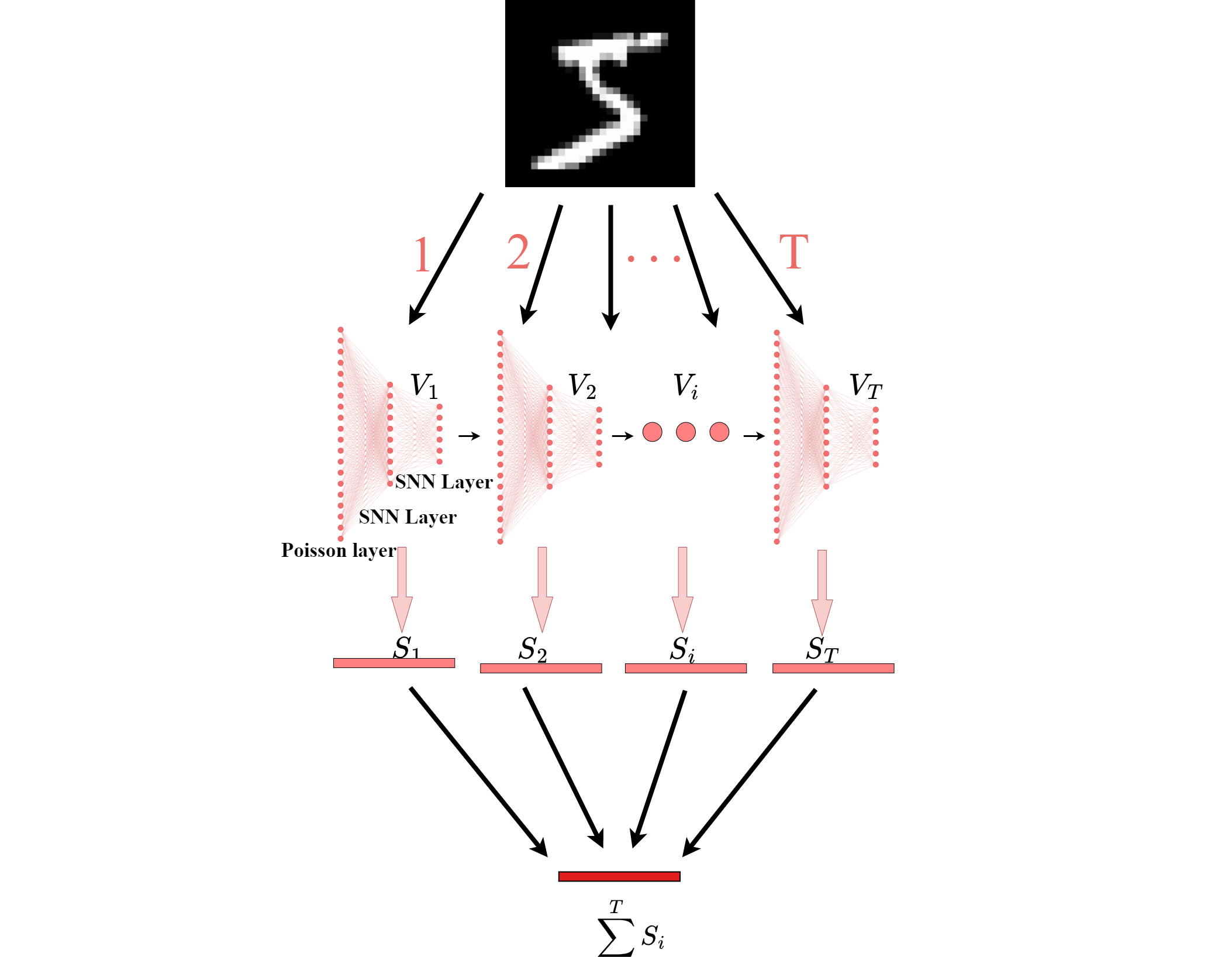}
}
\caption{(a)The circle represents the message space $Z_p$, while the numbers on the circumference represent the finite field of the ciphertext space conventionally denoted as $Z_q$. This mapping relationship between $Z_p$ and $Z_q$ is reflected in the partitioning of the circle. Operations on both plaintext and ciphertext follow the rules of finite fields, where values exceeding the bounds undergo modular arithmetic, wrapping around to the opposite end. (b)The output of an SNN corresponds to the firing frequency within a specific time window of the output layer, where the magnitude of the firing rate reflects the response strength towards a particular category. Thus, the network is required to operate for a designated duration, utilizing the average firing rate after $T$ time steps as the classification score. During each time step, the image sequentially traverses the Poisson layer, SNN hidden layer, and SNN output layer.}
\end{figure*}

Figure \ref{ZP} illustrates a finite field $Z_p$, which forms a cyclic group. Similar to a circular arrangement, if a number exceeds $p/2$, it wraps around to the opposite end, as well as values below $-\frac{p}{2}+1$. Being defined on $Z_p$, the intermediate values during DiSNN computation must be accommodated in $Z_p$, or else exceeding the boundaries will result in computational errors. This means that the inequality 
$
-\frac{p}{2}+1 \le \beta \textless \alpha \textless \frac{p}{2}
$
must be satisfied. However, large $p$ leads to a decrease in computational efficiency for homomorphic encryption. Therefore, selecting an appropriate $\tau$ that strikes a balance between computational efficiency and discretization accuracy becomes a crucial consideration.

\section{Homomorphic Evaluation of DiSNN}

\vspace{-1mm}This section will delve into a detailed analysis of how DiSNN performs prediction on homomorphically encrypted images. As the prediction procedure of SNN shown in Figure \ref{main_process_figure}, all operations performed on ciphertext can be summarized as \textbf{Fire}, \textbf{Reset}, and Multisum(scalar multiplication and summation). Poisson encoding merely involves comparing magnitudes, which can be computed using the \textbf{Fire} function. Multisum is naturally supported in FHE, so the challenges lie in performing \textbf{Reset} and \textbf{Fire} functions on ciphertext since they are non-polynomial functions. We leverage the programmable bootstrapping technique introduced by Chillotti et al.\cite{chillottiImprovedProgrammableBootstrapping2021} to execute the \textbf{Fire} and \textbf{Reset} functions of the SNN model while simultaneously refreshing the noise of the ciphertext.

\subsection{Homomorphic Computation of Multisum}
We select a neuron from the SNN layer, and its input is expressed as Equation \ref{couhang}, which is correct, as long as the noise carried by the ciphertext does not exceed the upper bound of the noise shown in the following Remark.

\begin{equation}
\label{couhang}
\begin{aligned}
\sum_j \hat{w}_{ij}LWE(S_j [t]) &= \sum_j LWE(\hat{w}_{ij} \cdot S_j [t]) \\
&= LWE(\sum_j \hat{w}_{ij} S_j [t]) \\
&= LWE(\hat{I}_{i} [t]).
\end{aligned} 
\end{equation}

\begin{remark}\label{crectmultisum}
It is observed that the multiplication and addition operations on ciphertexts will amplify the noise carried by the ciphertexts. To ensure the correctness of the above computation which equals to 
$$Dec(\sum_j \hat{w}_{ij}LWE(S_j [t])) = Dec(LWE(\hat{I}_{i} [t])).$$
There are two conditions that need to be satisfied: (1) $\sum \limits_j \hat{w}_{ij}S_j[t] \in [-\frac{p}{2},\frac{p}{2})$; (2) the noise does not grow beyond the noise bound. The first condition is easy to satisfy by choosing a sufficiently large message space $Z_p$. To address the noise issue, let us assume that $LWE(S_j[t])$ has an initial noise $\sigma$ (as each spike is generated via bootstrapping). After the multiplication and addition operations, the noise in the ciphertext grows to $|\sum\limits_j \hat{w}_{ij}|\cdot \sigma$, which is proportional to the discretization parameter $\tau$. One way to control the noise growth is to decrease $\tau$, which may lead to a decrease in accuracy. Another approach is to trade off the security level by reducing the initial noise $\sigma$, where increasing the dimension of the LWE problem $n$ could remedy the situation\cite{albrechtConcreteHardnessLearning2015}.
\end{remark}

\subsection{Homomorphic Computation of Fire Function}
The \textbf{Fire} function is a non-polynomial function, so we must rely on the Theorm \ref{bootthm} to evaluate it and refresh the ciphertext noise simultaneously. We propose a solution to implement the \textbf{Fire} function on ciphertexts, referred to as the \textbf{FHE-Fire} function, which can be realized as:
\begin{equation}
\begin{aligned}
\textbf{FHE-Fire}(LWE(m)) &\triangleq bootstrap(LWE(m))+1 \\
&= \begin{cases}
LWE(2), & \text{if } m \in [0,\frac{p}{2}), \\
LWE(0), & \text{if } m \in [-\frac{p}{2},0)
\end{cases} \\
&= LWE(2 \cdot Spike).
\end{aligned}
\end{equation}

by defining the program function $g$ of bootstrapping as:
\begin{equation}
g(m) \triangleq \begin{cases}
1, & \text{if } m \in [0,\frac{p}{2}), \\
-1, & \text{if } m \in [-\frac{p}{2},0).
\end{cases}    
\end{equation}

In this case, the spike signal is mapped to $\{0,2\}$, doubling its original value. This adds a slight complication for the subsequent fully connected layer's computation. However, it can be easily overcome. Since the spike signal is now doubled, we ensure the consistency of the final result by halving the weights as the following equation show: 
\begin{equation}
\begin{aligned}
LWE(\hat{I}) &= \sum_j \hat{w}_{ij} \cdot LWE(S_j)  \approx \sum_j \lfloor\frac{\hat{w}_{ij}}{2} \rfloor \cdot LWE(2 \cdot S_j) \\
&\approx \sum_j \text{Discretized}(w_{ij},\frac{\tau}{2}) \cdot LWE(2 \cdot S_j).
\end{aligned}    
\end{equation}

This approach also benefits the control of ciphertext noise by halving the discretization parameter $\tau$. Since ciphertext $LWE(2 \cdot S_j)$ obtained through bootstrapping has very low initial noise $\sigma$, our method reduces the noise by half, which can be easily proven using Remark \ref{crectmultisum}. This allows us to have lower noise growth and enables us to make more confident parameter choices.

\vspace{-0.5cm}
\subsection{Homomorphic Computation of Reset Function}
The \textbf{Reset} function describes two characteristics of the neuron model's membrane potential. First, when the membrane potential $V$ exceeds $V_{threshold}$, a spike is emitted and the membrane potential is back to $V_{reset}$. Second, the membrane potential cannot be lower than $V_{reset}$, so if such a value is generated during the computation process, it needs also to be set to $V_{reset}$.

For convenience, the $V_{reset}$ is often set to 0. For non-zero conditions, we can shift the $V_{reset}$ to 0 using a translation method.
Similarly to the \textbf{Fire} function, here we set the program function $g$ of bootstrapping as:


\begin{equation}
g(m) \triangleq
    \begin{cases}
        0, & \quad \text{if } m \in [\hat{V}_{threshold},\displaystyle\frac{p}{2}), \\
m, & \quad \text{if } m \in [0,\hat{V}_{threshold}), \\
0, & \quad \text{if } m \in [\hat{V}_{threshold}-\displaystyle\frac{p}{2},0), \\
-(m+\displaystyle\frac{p}{2}), & \quad \text{if } m \in [-\displaystyle\frac{p}{2},\hat{V}_{threshold}-\frac{p}{2}), 
    \end{cases}
\end{equation}

where $g(x) = -g(x+\frac{p}{2})$ must be satisfied. Then, the \textbf{FHE-Reset} function can be computed as follows:
\begin{equation}\label{FHE-Reset}
\begin{aligned}
\textbf{FHE-Reset}(LWE(m)) &\triangleq bootstrap(LWE(m) \\
&= \left\{ \begin{aligned} 
LWE(0)\ &,\ m \in [\hat{V}_{threshold},\frac{p}{2}), \\
LWE(m)\ &,\ m \in [0,\hat{V}_{threshold}), \\
LWE(0)\ &,\ m \in [\hat{V}_{threshold}-\frac{p}{2},0), \\
LWE(-(m+\frac{p}{2}))\ &,\ m \in [-\frac{p}{2},\hat{V}_{threshold}-\frac{p}{2}). \\
\end{aligned}
\right. 
\end{aligned}
\end{equation}

Note that $g$ does not match \textbf{Reset} function on interval $[-\frac{p}{2},\hat{V}_{threshold}-\frac{p}{2})$, which will lead to computation error. Therefore, the computation must be evaluated within the interval $[\hat{V}_{threshold}-\frac{p}{2},\frac{p}{2})$, equaling to $\hat{V}_{threshold} - \frac{p}{2} \leq \beta  \leq \alpha \le \frac{p}{2}$ from Proposition \ref{maxmin}. The given conditions can be simplified to $\alpha \textless \frac{p}{2}$, based on the insights provided by Equation \ref{maxminEq}. This condition serves as a necessary requirement for selecting the message space, as it ensures that the intermediate variable naturally falls within the correct range without the need for additional conditions.

Moreover, \textbf{FHE-Reset} function not only realizes computing the \textbf{Reset} function on the ciphertext but also serves the purpose of refreshing the noise that is accumulated during the computation of $LWE(\hat{V}[t])$. This means that we don't need to worry about the noise issue, and it can support an arbitrary number of computations.

\section{Experiment}
In this chapter, we will actually build a practical FHE-DiSNN network and demonstrate its advantages and limitations in experiments by comparing it with the original SNN. There are three parts in this section. first, we determine the simulation model and network depth of our network to train a convergent SNN. Second, we proceed to convert the well-trained SNN into DiSNN. Finally, in the third part, we conduct experiments to evaluate the accuracy and efficiency of FHE-DiSNN in performing forward propagation on encrypted images. This assessment provides insights into the performance of FHE-DiSNN in a secure and encrypted environment.

\subsection{Building an SNN in the clear}
We select a 784 ($28 \times 28$)-dimensional Poisson encoding layer as the input layer, and 30-dimensional and 10-dimensional IF model as the hidden layer and output layer, respectively. We utilize the Spikingjelly\cite{SpikingJelly} library, a PyTorch-based framework for SNN. Training SNN is an intriguing research direction, and in this study, the gradient surrogate method is chosen to train SNNs. Other commonly used training methods such as ANN-to-SNN conversion and unsupervised training with STDP are not extensively discussed here. In general, the network needs to run for a period of time, taking the average firing rate over $T$ time steps as the basis for classification.

The prediction process is essentially a forward propagation of the trained model, with the key difference that gradient approximation is not required. The predictive accuracy improves as $T$ increases, but the marginal effect leads to a diminishing rate of improvement in accuracy, while the time consumption continues to grow linearly. Hence, in order to maintain accuracy, it is vital to minimize the value of $T$ as much as possible. We are aware that in encrypted computations, bootstrapping is the most time-consuming operation. Here, we provide an estimate for the number of bootstrapping operations in FHE-DiSNN.

\begin{proposition}\label{timeconsume}
For a single-layer FHE-DiSNN network with n-dimensional input, k-dimensional hidden layer, m-dimensional output, and simulation time T, the number of required bootstrapping operations is $(n+2k+2m)T$.
\end{proposition}
\begin{proof}
In FHE-DiSNN, the Poisson encoding requires one \textbf{FHE-Fire} function call, and the discharge and reset processes of the model each require one \textbf{FHE-Fire} function and one \textbf{FHE-Reset} function call. In a single simulation time, there are $n$ Poisson encoding operations and $m+k$ \textbf{FHE-Fire} and \textbf{FHE-Reset} operations, resulting in a total of $(n+2m+2k)$ bootstrapping. With $T$ repeated simulation steps, the total number of bootstrapping is given by:
$$
\text{nums} = (n+2k+2m) \cdot T.
$$
\end{proof}

We can reduce the number of bootstrapping from two aspects to improve experimental efficiency. First, we can encrypt the message after Poisson encoding, which eliminates the requirement for $nT$ times bootstrapping.
Second, we can shorten the simulation time $T$. The curve in Figure \ref{Taccuracy} shows that the optimal trade-off is achieved at $T=10$, where the accuracy is comparable to the highest level achieved at $T=20$.

\vspace{-0.6cm}
\begin{figure*}[htb]
\centering
\subfigure [$T$-accuracy in SNN.]{
 \label{Taccuracy}
\includegraphics[width=0.475\columnwidth,height=0.2\textwidth]{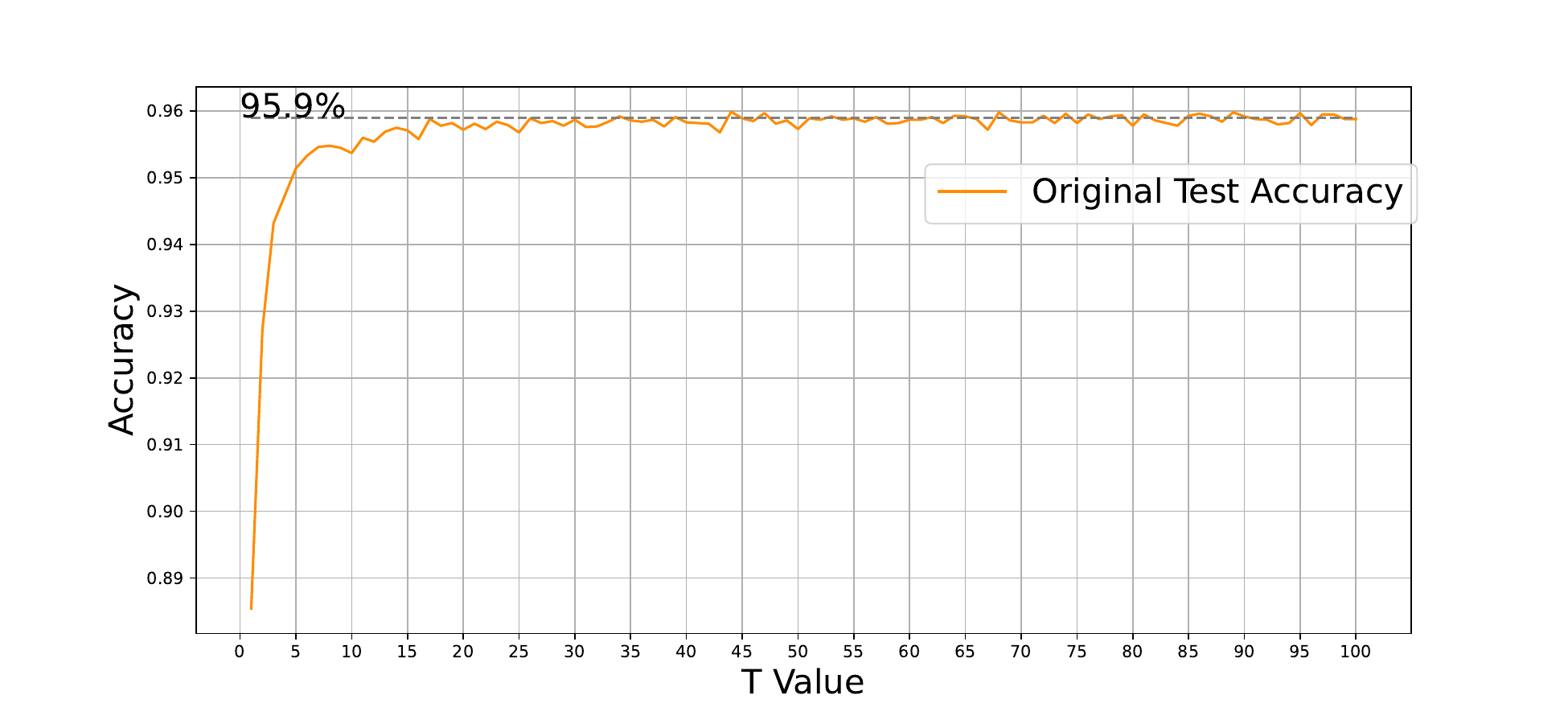}
}
\subfigure[$\tau$-accuracy in DiSNN] {
\label{tauacc}
\includegraphics[width=0.475\columnwidth,height=0.2\textwidth]{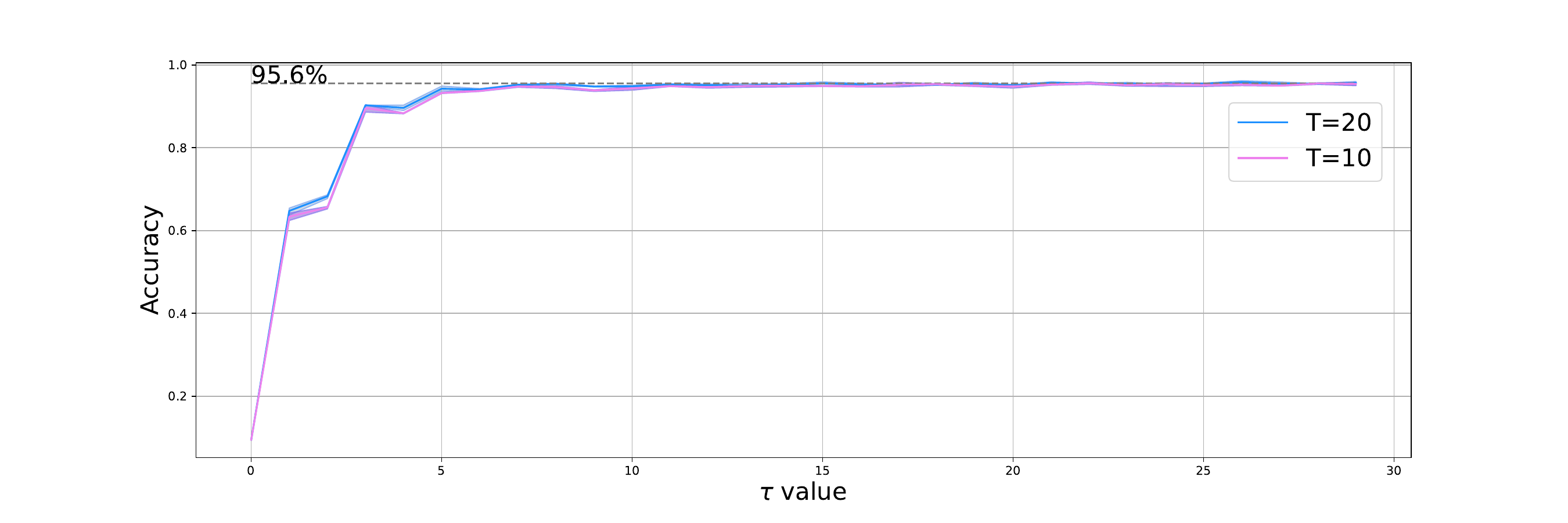}
}
\caption{(a)The curve (orange) illustrates the influence of simulation time $T$ on the prediction accuracy of the original SNN. By $T=10$, the SNN has reached an accuracy of 95.3\%, and at $T=20$, it has essentially achieved its maximum level of accuracy. (b)The graph illustrates the correlation between $\tau$ and the prediction accuracy in DiSNN. Note that the light blue and lavender shading in the figure represents the maximum and minimum fluctuation intervals of the results of the five independent experiments under the conditions of $T=20$ and $T=10$, respectively, while the blue and purple lines represent the average values of the experimental results.}
\label{figure}
\end{figure*}

\vspace{-0.5cm}
\subsection{Constructing an FHE-DiSNN from SNN }

The accuracy of DiSNN improves with the increase of $\tau$, but the marginal effect is also present. Moreover, the increase in $\tau$ leads to linear growth of noise, resulting in elevated computational costs. Therefore, selecting an appropriate value for $\tau$ is crucial. Following the design of the FHE-DiSNN algorithm, we conduct experiments to show the relationship between $\tau$ and the prediction accuracy of DiSNN for $T=10$ and $T=20$. The curve is plotted in Figure \ref{tauacc}, where the result indicates that $\tau=10$ and $\tau=10$ achieve the optimal trade-off and the highest accuracy, respectively.

The next step is to choose an appropriate $p$ of the message space. According to Equation \ref{maxminEq}, the requirements can be simplified to$\alpha \textless \frac{p}{2}$. The maximum value depends on $\max\limits_{i,j}(\sum\limits_jw_{ij}S_j[t])$, which is a fixed value for a well-trained network and can be pre-determined. Through several experiments, we have $\alpha \approx 50\tau$ for the first layer and approximately $10\tau$ for the second layer. A message space size of $p=1024$ is enough to accommodate the DiSNN with $\tau=10$, and $p=2048$ for $\tau=20$. For our experiments, we have selected the STD128 parameter set\cite{micciancioBootstrappingFHEWlikeCryptosystems2020} shown as follows:\\
-Size of message space: $p =1024$ or $2048$\\
-Dimension of LWE ciphertext: $n=512$\\
-Degree of the polynomials in the ring: $N=1024$\\
-Bits of the decomposition basis of KeySwitching: $B_{ks} = 14$\\
-Bits of the decomposition basis for TGSW ciphertext: $B_g = 7$\\

\vspace{-0.5cm}
\subsection{Exhibiting experiment result}
We conduct the following experimental process on an Intel Core i7-7700HQ CPU @ 2.80 GHz:\\
\textbf{1.} The grayscale handwritten digit image is encoded by the Poisson encoding layer.\\
\textbf{2.} The Poisson-encoded image is then encrypted into LWE ciphertexts.\\
\textbf{3.} The ciphertext is multiplied by the discretized plaintext weights and passed into the SNN layer.\\
\textbf{4.} The IF neuron model in the SNN layer calculates the charging, firing, and reset procedure on the ciphertext. The bootstrapping operations involved in this process are accelerated using FFT.\\
\textbf{5.} Steps \textbf{1-4} repeat $T$ times, and the resulting outputs are accumulated as classification scores.\\
\textbf{6.} Decrypt and the highest score are selected as the classification result.\\
In Table \ref{result}, we show the experimental results, with the \textbf{Fire} and \textbf{Reset} functions in Step 4 implemented according to our method in Section 4.


\begin{table}[htbp]
\centering
\vspace{-0.5cm}
\begin{tabular}{|c|c|c|c|c|c|}

\hline$(\mathrm{T}, \tau)$ & FHE-DiSNN & DiSNN &  SNN & Time per step & Time per image \\
\hline$(10,10)$ & $94.40 \%$ & $95.00 \%$ & $95.30 \%$ & 0.83s & 8.31s\\
$(10,20)$ & $94.40 \%$ & $95.00 \%$ & $95.30 \%$  & 0.86s & 8.67s\\
$(20,10)$ & $94.80 \%$ & $95.10 \%$ & $95.70 \%$ & 0.81s & 16.21s\\
$(20,20)$ & $95.10 \%$ & $95.30 \%$ & $95.70 \%$ & 0.79s & 15.97s\\
\hline\hline
$ $ & FHE-DiNN & DiNN & NN & Time per step & Time per image\\
\hline
30 neurons & $93.46\%$ & $93.55 \%$ & $94.46 \%$ & 0.49s & 0.49s\\

\hline
\end{tabular} 
\vspace{0.15cm}
\caption{The experiment result. The table above presents the results of experiments conducted with four different parameter sets. The first three columns represent the prediction accuracy of FHE-DiSNN, DiSNN, and SNN, respectively. In the case of the original SNN, only the parameter $T$ influences the prediction performance. The fourth and fifth columns display the time consumed by FHE-DiSNN during a single time step and a complete prediction procedure, which, as stated by Proposition \ref{timeconsume}, is directly proportional to the simulation time $T$. The last two lines excerpt the experimental results of FHE-DiNN\cite{bourseFastHomomorphicEvaluation2018} for comparison, and SNN and NN have the same structure. }
\label{result}
\end{table}

The results reveal, as emphasized at the beginning of the article, that discretization has minimal impact on SNN. The original SNN, with 30 hidden layers, achieves a stable accuracy rate of around 95.7\%, outperforming many second-generation neural networks. DiSNN also demonstrates a commendable accuracy rate of 95.3\%(the best parameters). Importantly, it only incurs a loss of 0.4\% compared to the original SNN, showcasing the inherent advantages of DiSNN. Furthermore, FHE-DiSNN performs impressively, consistently surpassing 94\% accuracy across four parameter sets. Particularly, the (20,20) parameter set demonstrates comparable performance to DiSNN. However, FHE-DiSNN suffers from time inefficiency due to the increased number of bootstrapping operations caused by the simulation time $T$, resulting in each prediction taking 8(16) seconds with 0.8 seconds consuming on average for a single simulation step.

During the experimental process, we observe that the number of spike firings differs between the FHE-DiSNN and DiSNN during computation. This suggests that certain ciphertexts may encounter noise overflow. However, this has a minimal effect on the final categorization outcomes. This is because slight noise overflow only causes a deviation of $\pm 1$, and abnormal spike firings occur only when the value is at the edge of $\hat{V}_{\text{threshold}}$, with a small probability. Additionally, individual instances of abnormal spike firings are effectively mitigated within the $T$ simulation time. This indicates that FHE-DiSNN exhibits a considerable level of tolerance toward the noise, which is a highly intriguing experimental finding.

\section{Conclusion}
This paper serves as an initial exploration of the fusion of SNN(Spiking Neural Networks) with homomorphic encryption and presents a research avenue brimming with boundless potential. This innovation facilitates us in terms of both low energy consumption from the machine learning side and data privacy from a security point of view. We offer an estimation of the maximum upper bound for discretization error in DiSNN and substantiate its expected error to be $\lambda/4$ from a mathematical expectation perspective. Experimental results further validate this finding. Furthermore, we leverage TFHE bootstrapping to construct \textbf{FHE-Fire} and \textbf{FHE-Reset} functions, enabling support for computations of unlimited depth. Besides, our proposed framework is easy to scalable and extended to more complicated neural models. 

However, there still remain some challenging tasks for further research, such as more complex neuron equations, different encoding methods, parallel computing, and so on. Besides, as highlighted by Proposition \ref{timeconsume}, Poisson encoding introduces numerous bootstrapping operations (equal to the dimension of input data), which can have a high evaluation time. And this is also one of our future directions.

\bibliographystyle{splncs04}
\bibliography{main}

\begin{appendix}

\section*{Appendix A: SNN}\label{AppendixA}
\subsection*{Encoding Strategies}\label{A.1}
To handle various stimulus patterns, SNNs often use abundant coding methods to process the input stimulus. At present, the most common neural coding methods mainly include rate coding, temporal coding, bursting coding and population coding \cite{Georgopoulos1986NeuronalPC}. For visual recognition tasks, rate coding is a popular scheme. Rate coding\cite{Wu2017SpatioTemporalBF,Zhan2021EffectiveTL} is mainly based on spike counting, and Poisson distribution can describe the number of random events occurring per unit of time.

In our study, the inputs are encoded into rate-based spike trains by the Poisson process. Given a time interval $\Delta t$ in advance, then the reaction time is divided into $T$ intervals evenly. During each time step t, a random matrix $M_t$ is generated using uniform distribution in $[0,1]$. Then, we compare the original normalized pixel matrix $X_o$ with $M_t$ to determine whether the current time $t$ has a spike or not. The final encoding spike train $X$ is calculated by using the following equation:
\begin{equation*}
   X(i,j)= 
   \begin{cases}
   0, &\ X_o(i,j) \leq M_t(i,j)\\
   1, &\ X_o(i,j) > M_t(i,j)
   \end{cases}
\end{equation*}
where $i$ and $j$ are the coordinates of the pixel points in the
images. In this way, the encoded spike trains follow the Poisson distribution.

\subsection*{Training methods}\label{A.2}
Supervised learning algorithms for deep spiking neural networks mainly include indirectly supervised learning algorithms represented by ANN-converted SNN and directly supervised algorithms represented by spatiotemporal back-propagation. However, the lack of differentiability of the spike function is still a problem we have to confront. At present, a common solution is to use a continuous function similar to it to replace the spike function or its derivative, which is called surrogate gradient method, resulting in a spike-based BP algorithm. Wu et al.\cite{Wu2017SpatioTemporalBF} introduce four curves to 
approximate the derivative of spike activity denoted by $f_1, f_2, f_3, f_4$ as follow. 
$$
\begin{aligned}
& f_1(V)=\frac{1}{a_1} \operatorname{sign}\left(\left|V-V_{th}\right|<\frac{a_1}{2}\right), \\
& f_2(V)=\left(\frac{\sqrt{a_2}}{2}-\frac{a_2}{4}\left|V-V_{t h}\right|\right) \operatorname{sign}\left(\frac{2}{\sqrt{a_2}}-\left|V-V_{t h}\right|\right), \\
& f_3(V)=\frac{1}{a_3} \frac{e^{\frac{V_{th}-V}{a_3}}}{\left(1+e^{\frac{V_{th}-V}{a_3}}\right)^2}, \\
& f_4(V)=\frac{1}{\sqrt{2 \pi a_4}} e^{-\frac{\left(V-V_{t h}\right)^2}{2 a_4}},
\end{aligned}
$$

In general, SNN training follows the three principles: (1)The output of spiking neurons is binary and can be affected by noise. The firing frequency over time is used to represent the response strength of a category for classification. (2)The goal is for only the correct neuron to fire at the highest frequency while others remain silent. MSE loss is often used and has shown better results. (3)The network state needs to be reset after each simulation.
It is important to note that training such an SNN requires a linear amount of memory and simulation time $T$. A larger $T$corresponds to a smaller simulation time step and more "fine-grained" training, but it does not necessarily result in better training performance. When $T$ is too large, the SNN unfolds into a very deep network in terms of time, which can lead to gradient vanishing or exploding during gradient computation.

\vspace{0.2cm}
\begin{minipage}[c]{0.5\textwidth}
		\centering
		\setlength{\tabcolsep}{0.5 mm}
		\scalebox{0.9}{
    \begin{tabular}{l}
    \toprule
    Load the train data to traindata\_loader \\
    \bf{for img, label in traindata\_loader:}\\
    \quad  label\_onehot = onehot\_encode(label) \\
    \quad \bf{for 1:T do} \\
    \quad \quad img\_encode = Poisson\_Encode(img) \\
    \quad \quad I = SNN.Linea(img\_encode) \\
    \quad \quad spike = SNN.Atan(I) \\
    \quad \quad output += spike \\
    \quad output /= T \\
    \quad loss = MSEloss(output, label\_onehot) \\
    \quad loss.backward() \\
    \quad optimizer.step() \\
    \bottomrule
    \end{tabular}
	}
    \captionof{table}{Train Process}
    \label{tab_train_process}

\end{minipage}
\begin{minipage}[c]{0.5\textwidth}
		\centering
		\setlength{\tabcolsep}{0.5mm}
		\vspace{-1mm}
		\scalebox{0.85}{\begin{tabular}{l}
    \toprule
    Load the test data to testdata\_loader \\
    \bf{for img in testdata\_loader:}\\
    \quad \bf{for 1:T do} \\
    \quad \quad img\_encode = Poisson\_Encode(img) \\
    \quad \quad I = SNN.Linea(img\_encode) \\
    \quad \quad spike = SNN.Atan(I) \\
    \quad \quad output += spike \\
    \quad result = argmax(output) \\
    \bottomrule
    \end{tabular}}
    \captionof{table}{Test Process}
    \label{tab_test_process}
\end{minipage}

Table\ref{tab_train_process} and Table\ref{tab_test_process} present the pseudocode for the training and testing processes, respectively, providing a more intuitive representation of their distinctions.

The longer the simulation time $T$, the higher the testing accuracy. However, increasing $T$ results in a significant increase in time consumption, especially when dealing with encrypted data.

\section*{Appendix B: Homomorphic Encryption shemes} \label{HomomorphicEncryption}
\noindent\textbf{LWE} We recall that a LWE(Learning With Errors)\cite{regevLatticesLearningErrors2005} ciphertext encrypting $m \in Z_p$ has the form:\\
$$LWE_s(m)=(\mathbf{a}, b)=\left(\mathbf{a},\langle\mathbf{a}, \mathbf{s}\rangle+e+ \left\lfloor\frac{q}{p} m\right\rceil\right) \bmod q $$
here, $\mathbf{a} \in \mathbb{Z}_q^n$ and $b \in \mathbb{Z}_q$, and the keys are vectors $\mathbf{s} \in \mathbb{Z}_q^n$. The ciphertext $(\mathbf{a}, b)$ is decrypted using:
$$
\left\lfloor \frac{p}{q}(b-\langle\mathbf{a}, \mathbf{s}\rangle ) \right\rceil \bmod p=\left\lfloor m+\frac{p}{q} e\right\rceil=m
.$$

\noindent\textbf{Note}. The function $\left\lfloor \cdot \right\rceil$ has the nice property: if $m_0 \equiv m_1 \bmod p$, then 
$
\left\lfloor\frac{q}{p} m_0\right\rceil \equiv \left\lfloor\frac{q}{p} m_1\right\rceil \bmod q
$
, so the operations on message polynomials mod $p$ are realized when computing ciphertext modulo $q$.
\\
\\
\noindent\textbf{RLWE}
 An RLWE(Ring Learning With Errors)\cite{lyubashevskyIdealLatticesLearning2012} ciphertext of a message $m(X) \in R_{p,N}$ can be obtained as follows:
$$R L W E_s(m(X)) =  \left(a(X), b(X)\right), \text{where } b(X) = a(X) \cdot s(X)+e(X)+ \left\lfloor \frac{q}{p} m(X) \right\rceil,$$
where $a(X) \leftarrow R_{q,N}$ is chosen uniformly at random, and $e(X) \leftarrow \chi_\sigma^n$ is chosen from a discrete Gaussian distribution with parameter $\sigma$. The decryption algorithm for RLWE is similar to LWE.
\\
\\
\noindent\textbf{GSW}
GSW\cite{gentryHomomorphicEncryptionLearning2013} is recognized as the first third-generation FHE scheme, and in practice, the RGSW\cite{gentryHomomorphicEncryptionLearning2013} variant is commonly used. Given a plaintext $m \in \mathbb{Z}_p$, the plaintext $m$ is embedded into a power of a polynomial to obtain $X^m \in R_{p,N}$, which is then encrypted as $RGSW(X^m)$. RGSW enables efficient computation of homomorphic multiplication, denoted as $\diamond$, while effectively controlling noise growth:
$$
\begin{aligned}
& R G S W\left(X^{m_0}\right) \diamond R G S W\left(X^{m_1}\right)=R G S W\left(X^{m_0+m_1}\right),\\
& R L W E\left(X^{m_0}\right) \diamond R G S W\left(X^{m_1}\right)=R L W E\left(X^{m_0+m_1}\right).
\end{aligned}
$$
Note that the first multiplication involves two RGSW ciphertexts, whereas the second multiplication operates on an RLWE ciphertext and an RGSW ciphertext and is used in bootstrapping.

\end{appendix}

\end{document}